\documentclass[lnicst,sechang,a4paper]{svmultln}
\usepackage{color}
\usepackage{leftidx}
\usepackage{cite}
\usepackage{microtype}
\usepackage{cleveref}
\usepackage{graphicx,amssymb,amstext,amsmath}
\usepackage[ruled,vlined,lined,ruled,linesnumbered]{algorithm2e}
\usepackage{float}

\usepackage{makeidx}  
\begin{document}
\mainmatter              
\title{Backscatter-assisted Relaying in Wireless Powered Communications Network}
\titlerunning{ }  
%
\author{Yuan Zheng \and Suzhi Bi
 \and Xiaohui Lin}
\authorrunning{ }   
%
\tocauthor{Suzhi Bi, Xiaohui Lin}

\institute{College of Information Engineering, Shenzhen University,\\
Shenzhen,  Guangdong, 518060, China\\
\email{zhengyuan2016@email.szu.edu.cn,{bsz,xhlin}@szu.edu.cn}}

\maketitle

\begin{abstract}        
This paper studies a novel cooperation method in a  two-user wireless powered communication network (WPCN), in which one hybrid access point (HAP) broadcasts wireless energy to two distributed wireless devices (WDs), while the WDs use the harvested energy to transmit their independent information to the HAP. To tackle the user unfairness problem caused by the near-far effect in WPCN, we allow the WD with the stronger WD-to-HAP channel to use part of its harvested energy to help relay the other weaker user's information to the HAP. In particular, we exploit the use of backscatter communication during the wireless energy transfer phase such that the helping relay user can harvest energy and receive the information from the weaker user simultaneously. We derive the maximum common throughput performance by jointly optimizing the time duration and power allocations on wireless energy and information transmissions. Our simulation results demonstrate that the backscatter-assisted cooperation scheme can effectively improve the throughput fairness performance in WPCNs.
\end{abstract}

\section{Introduction}
Wireless communication is fundamentally constrained by the limited battery life of wireless devices. Frequent battery replacement/recharging will interrupt wireless communication and degrade the quality of communication service. Alternatively, radio frequency (RF) enabled wireless energy transfer (WET) technology can supply continuous and sustainable energy to remote WDs. Its application in wireless communication introduces a new networking paradigm, named wireless powered communication network (WPCN). Recent studies have shown that its deployment can largely reduce the network operational cost, and effectively improve the communication performance, e.g., achieving longer operating time and more stable throughput \cite{2014:Bi,2016:Bi1,2014:Ju1,2016:Bi2,2017:Bi,2018:Bi,2015:Bi}. For example, \cite{2014:Ju1} proposed a harvest-then-transmit protocol in WPCN, where one hybrid access point (HAP) with single antenna first transfer RF energy to all WDs in the downlink (DL), and then the WDs transmit information to the HAP in the uplink (UL) using their received energy in a time-division-multiple-access (TDMA) manner. It is observed in \cite{2014:Ju1} that the WPCN suffers from a doubly near-far problem among WDs in different locations, where a far user from the HAP achieves low throughput because they receive less energy and need more power to transmit information. To solve the doubly near-far problem and improve user fairness, several different user cooperation schemes have been proposed \cite{2014:Ju2,2017:MM,2017:Yuan,2017:Wu}. For instance, \cite{2014:Ju2} proposed a two-user cooperation, where the near user helps relay the far user's information to the HAP. \cite{2017:MM} allows two cooperating users to from a distributed virtual antenna array. \cite{2017:Yuan} considered a cluster-based user cooperation, where a multi-antenna HAP applies WET to power a cluster of remote WDs and receives their data transmissions.

 A major design issue of the existing user cooperation schemes is that the overhead (both energy and time) consumed on information exchange between the collaborating users. Alternatively, the recent development of ambient backscatter (AB) communication provides an alternative to reduce such collaborating overhead. Specifically, AB enables a WD to transmit information passively to another device in the vicinity by backscattering the RF signal in the environment, e.g., WiFi and cellular signals, thus achieving device battery conservation. Several recent studies have devoted to improve the data rate of AB, such as proposing new signal detection method and AB communication circuit designs \cite{2013:Vli,2016:Gw}. However, the performance of conventional ambient backscatter communication greatly depends on the conditions of time-varying ambient RF signal, which is not controllable in either its strength or time availability.
\begin{figure}
\vspace{-0.8cm}
\setlength{\abovecaptionskip}{-0.2cm}   %
\setlength{\belowcaptionskip}{-1cm}
  \centering
   \begin{center}
      \includegraphics[width=0.7\textwidth]{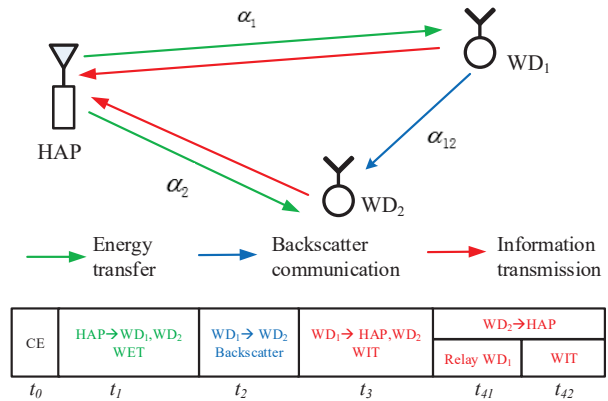}
   \end{center}
  \caption{A two-user WPCN and transmission protocol for user cooperation. }
  \label{Fig.1}
\end{figure}

In this paper, we consider a novel user cooperation method in WPCN which uses backscatter communication. As shown in Fig.~\ref{Fig.1}, we consider two wirelessly powered WDs that harvest RF energy in the DL and transmit cooperatively their information to the HAP in the UL. Unlike in conventional cooperation in WPCN where one WD transmits its information actively to the helping WD, we reuse the WET signal for achieving simultaneous information transmission in a passive manner. This largely saves the collaborating overhead. Besides, compared to conventional AB communication, the use of WET is fully controllable in the RF signal strength and transmission time. With the proposed backscatter-assisted cooperation method, we formulate a rate optimization problem that maximizes the minimum throughput between the two WDs, by jointly optimizing the system transmit time allocation and the power allocations of energy-constrained WDs. Efficient algorithm is  proposed to solve the optimization optimally. Simulation results show that, compared to conventional cooperation based on active communication, the proposed passive cooperation can effectively enhance the throughput performance of energy-constrained devices in WPCN.

\section{System Model}
\subsection{Channel Model}

As show in Fig.~\ref{Fig.1}, we consider a WPCN consisting of one HAP and two users denoted by WD$_1$ and WD$_2$, where the WDs harvest RF energy in the DL and transmit wireless information in the UL. It is assumed that each device is equipped with one antenna and both WET and WIT operate over the same frequency band. We assume that the channel reciprocity holds between the DL and UL, the channel coefficient between the HAP and WD$_i$ is denoted as $\alpha_i$ and the channel power gain is denoted as $h_i=|\alpha_i|^2, i=1,2$. Besides, the channel coefficient between WD$_1$ and WD$_2$ is denoted as $\alpha_{12}$ with the channel power gain $h_{12}=|\alpha_{12}|^2$. We assume without loss of generality that WD$_2$ has a better WD-to-HAP channel than WD$_1$, so which acts as a relay to forward the message of WD$_1$ to the HAP.
\begin{figure}
\vspace{-0.6cm}
\setlength{\abovecaptionskip}{-0.2cm}
\setlength{\belowcaptionskip}{-1cm}
  \centering
   \begin{center}
      \includegraphics[width=0.7\textwidth]{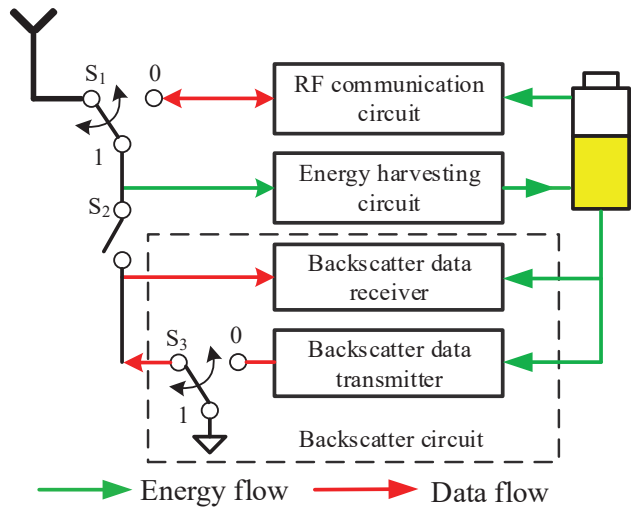}
   \end{center}
  \caption{Circuit block diagram of backscatter wireless user. }
  \label{Fig.2}
\end{figure}

In this paper, we consider that the two WDs can transmit information in both active (RF communication) and passive modes (backscatter communication). The circuit block diagram of two users is shown in Fig.~\ref{Fig.2}. With the two switches $S_1$ and $S_2$, the two WDs can switch flexibly among three operating mode as follows.
\begin{enumerate}
  \item \emph{Backscatter Mode} ($S_1=1$ and S$_2$ is closed): in this case, the antenna is connected to backscatter communication and energy harvesting circuits. A WD transmits information passively by backscattering the received RF signal. Specifically, a WD transmits ``1" or  ``0" by switching S$_3$ between reflecting or absorbing state, respectively. Accordingly, a backscatter receiver uses non-coherent detection techniques, e.g., energy detector\cite{2012:MAK}, to decode the transmitted bit. Notice that the energy consumption on the operation of backscatter transmitter can be well neglected due to the harvested energy during the absorbing state \cite{2017:DTH}.
  \item \emph{RF Communication Mode} ($S_1=0$): the antenna is connected to the RF communication circuit and the user can transmit or receive information using conventional RF wireless communication techniques. Here, the transmission energy consumption is supplied by the RF energy harvested from the HAP.
  \item \emph{Energy-harvesting Mode} ($S_1=1$ and S$_2$ is open): the antenna is connected to the energy harvesting circuit, which can convert the received RF signal to DC energy and store in a rechargeable battery. The energy is used to power the operations of all the other circuits.
\end{enumerate}

\subsection{Protocol Description}
  As shown in Fig.~\ref{Fig.1}, channel estimation (CE) is first performed with a fixed duration $t_0$, such that a central control point (such as the HAP) is aware of the channel coefficients $\{\alpha_1,\alpha_2,\alpha_{12}\}$. After CE, the system operates in four phases. In the first phase of duration $t_1$ the HAP broadcasts wireless energy in the DL with fixed transmit power $P_1$, while both the WDs harvest RF energy. In the second phase of duration $t_2$, the HAP continues to broadcast energy while WD$_1$ uses its backscatter communication circuit to transmit its information to the WD$_2$. Here, we assume the HAP neglects the backscattered signal due to the hardware constraint. Then, in the third phase, WD$_1$ operates in the conventional RF communication mode to transmit its information, using the harvested energy to WD$_2$. Notice that the HAP can overhear the RF transmission of WD$_1$ during this phase. In the last phase of length $t_4$, WD$_2$ first relays the user WD$_1$'s information to the HAP with average power $P_{41}$ over $t_{41}$ amount of time, and then transmits its own information to the HAP using its harvested energy with average power $P_{42}$ over $t_{42}$ amount of time, respectively, where $t_4=t_{41}+t_{42}$. Notice that we have a total time constraint
\begin{equation}
\label{t}
\small
t_0+t_1+t_2+t_3+t_{41}+t_{42}\leq T.
\end{equation}
For convenience, we assume $T=1$ in the sequel without loss of generality.
\section{Throughput Performance Analysis}
During the DL phase, the HAP transmits energy signal with the fixed power $P_1$ in $t_1$ amount of time. It is assumed that the energy harvested from the receiver noise is negligible. Hence, the amount of energy harvested by WD$_1$ and WD$_2$ can be expressed as \cite{2016:Bii}
\begin{equation}\label{energy}
  E_1^{(1)}={\eta}{t_1}{P_1}{h_1},\  E_2^{(1)}={\eta}{t_1}{P_1}{h_2},
\end{equation}
where $0\textless \eta \textless 1$ denotes the energy harvesting efficiency assumed fixed and equal for each user.

In the second stage of duration $t_2$, WD$_1$ uses backscatter communication to transmit its information to WD$_2$. Let $x_2(t)$ denote the transmitted energy signal by the HAP with $E[|x_2(t)|^2] =1$. We assume a fixed backscattering data rate $R_b$ bits/second, thus the duration of transmitting a bit is $1/R_b$ second. In particular, when WD$_1$ transmits a bit ``0'', WD$_2$ receives only the energy signal from the HAP
\begin{equation}
y_{2,0}^{(2)}(t) = \alpha_2 \sqrt{P_1} x_2(t) + n_2^{(2)}(t).
\end{equation}

Otherwise, when WD$_1$ transmits a bit ``1'', the received signal at WD$_2$ is a combination of both the HAP's energy signal and the reflected signal from WD$_1$, where
\begin{equation}
y_{2,1}^{(2)}(t) =\alpha_2 \sqrt{P_1}x_2(t) +  \mu\alpha_1\alpha_{12}\sqrt{P_1}x_2(t) + n_2^{(2)}(t),
\end{equation}
where $\mu$ denotes the signal attenuation coefficient due to the reflection at WD$_1$, $n_2^{(2)}(t)$ denotes the receiver noise at WD$_2$.
\begin{figure}
\vspace{-0.4cm}
\setlength{\abovecaptionskip}{-0.2cm}
\setlength{\belowcaptionskip}{-1cm}
  \centering
   \begin{center}
      \includegraphics[width=0.6\textwidth]{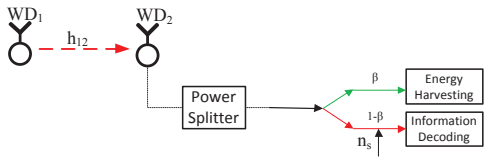}
   \end{center}
  \caption{Power Splitting Scheme in backscattering stage. }
  \label{Fig.3}
\end{figure}

Specifically, as shown in Fig.~\ref{Fig.3}, we apply a power splitting scheme, where the received RF signal is split into two parts. We denote $\beta\in[0,1]$ as the splitting factor, such that $\beta$ part of the signal power is harvested by the device, while the rest $(1-\beta)$ of the signal power is used for information decoding (ID). For simplicity, $\beta$ is assumed a constant in this paper. The information decoding circuit introduces an additional noise $n_s(t)$, which is assumed independent of the antenna noise $n_2(t)$. Thus, the signal at energy decoder and information decoder can be expressed as
\begin{equation}
y_{2,E}^{(2)}(t)=\sqrt{\beta}y_{2}^{(2)}(t),
y_{2,I}^{(2)}(t)=\sqrt{1-\beta}y_{2}^{(2)}(t)+n_s(t),
\end{equation}
where $y_2^{(2)}(t) = y_{2,0}^{(2)}(t)$ when sending ``0" and  $y_2^{(2)}(t) = y_{2,1}^{(2)}(t)$ when sending ``1". Without loss of generality, we assume that ``0" and ``1" are transmitted with equal probability. The harvested energy by WD$_2$ can be expressed as
\begin{equation}
E_2^{(2)} =\eta t_2\beta \frac{1}{2}((E[|y_{2,0}^{(2)}(t)|^2] +E[|y_{2,1}^{(2)}(t)|^2] ))
=\frac{1}{2}\eta t_2\beta  P_1(h_2+|\alpha_2+ \mu \alpha_1 \alpha_{12}|^2).
\end{equation}
Here, we assume that the signals received directly from the HAP and that reflected from WD$_1$ are uncorrelated due the the random phase change during backscatter. Meanwhile, we assume that WD$_1$ maintains its battery level unchanged during the backscatter stage, where the small amount of energy harvested is consumed on the on-off operations of the backscatter switches.

We denote the sampling rate of WD$_2$'s backscatter receiver as $NR_b$, such that it takes $N$ samples during the transmission of a bit, either ``0'' or ``1''. The following lemma derives the bit error probability (BER) of using an optimal energy detector to decode the received one-bit information.

\begin{lemma}
 The BER $\epsilon$ using an optimal energy detector for the backscatter communication is
 \begin{equation}
 \label{Pb}
\epsilon={\frac{1}{2}}erfc[{\frac{(1-\beta)P_1\mu^2 h_1h_{12}\sqrt{N}}{4((1-\beta)N_0+N_s)}}].
\end{equation}
\end{lemma}
\begin{proof}
 Due to the page limitation, the derivation is omitted here.
 \end{proof}

Then, the communication can be modeled as a binary symmetric channel, whose capacity (in bit per channel use) can be expressed as
\begin{equation}
C = 1+\epsilon log\epsilon+(1-\epsilon )log(1-\epsilon).
\end{equation}
Accordingly, the effective data rate from WD$_1$ to WD$_2$ is
\begin{equation}
R_1^{(1)}(\mathbf{t}) = C R_b t_2.
\end{equation}

 Within the sequel $t_3$ amount of time, WD$_1$ uses the harvested energy to actively transmit its information. By exhausting its harvested energy on WIT, the average transmit power of WD$_1$ is given by
 \begin{equation}
 \label{P3}
P_3=E_1^{(1)}/{t_3}={\eta}{P_1}{h_1}{t_1}/{t_3}.
\end{equation}
We denote $x_3(t)$ as the complex base-band signal transmitted by WD$_1$ with $E[|x_3(t)|^2] =1$. The received signals at WD$_2$ and the HAP in this time slot are expressed as
\begin{equation}
\label{y203}
y_2^{(3)}(t)=\alpha_{12}\sqrt{P_3}{x_3(t)}+n_2^{(3)}(t),
y_0^{(3)}(t)=\alpha_1\sqrt{P_3}{x_3(t)}+n_0^{(3)}(t),
\end{equation}
where $n_2^{3}(t)$ and $n_0^{(3)}(t)$ denote the receiver noises.

During the last time slot of duration $t_4$, the WD$_2$ first relays WD$_1$'s message to the HAP and then transmits its own message. Specifically, we denote the transmit power and time for relaying WD$_1$'s message as $P_{41}$ and $t_{41}$, and those for transmitting its own message as $P_{42}$ and $t_{42}$. Then, the total energy consumed by WD$_2$ is constrained as
\begin{equation}
\label{con}
t_{41}P_{41}+t_{42}P_{42}\leq E_2^{(1)}+E_2^{(2)}.
\end{equation}

Denote the time allocations as $\mathbf{t}=[t_1,t_2,t_3,t_{41},t_{42}]$, and the transmit power values $\mathbf{P}=[P_1,P_2,P_3,P_{41},P_{42}]$. For simplicity of illustration, we assume that the receiver noise power is $N_0$ at all receiver antennas except for the additional noise $n_s(t)$ introduced in the power splitter, whose power equals to $N_s$. Then, let $R_1^{(2)}(\mathbf{t},\mathbf{P}),R_1^{(3)}(\mathbf{t},\mathbf{P})$ and $R_1^{(4)}(\mathbf{t},\mathbf{P})$ denote the achievable rates of transmitting WD$_1$'s message from WD$_1$ to WD$_2$, from WD$_1$ to the HAP, and to the HAP relayed by WD$_2$, respectively, which are given by
\begin{equation}
\label{Ry12}
R_1^{(2)}(\mathbf{t},\mathbf{P})=t_3\log_{2}\left(1+\frac{P_3h_{12}}{N_0}\right),
R_1^{(3)}(\mathbf{t},\mathbf{P})=t_3\log_{2}\left(1+\frac{P_3h_1}{N_0}\right),
\end{equation}
\begin{equation}
\label{Ry14}
R_1^{(4)}(\mathbf{t},\mathbf{P})=t_{41}\log_{2}\left(1+\frac{P_{41}h_2}{N_0}\right).
\end{equation}

Thus, the achievable rate of WD$_1$ within the time slot of length $T=1$ can be expressed as\cite{2014:Ju2}
\begin{equation}\label{11}
 R_1(\mathbf{t},\mathbf{P}) = \min [R_1^{(1)}(\mathbf{t})+R_1^{(2)}(\mathbf{t},\mathbf{P}),R_1^{(3)}(\mathbf{t},\mathbf{P})+R_1^{(4)}(\mathbf{t},\mathbf{P})],
\end{equation}
and the achievable rate of WD$_2$ is
\begin{equation}
\label{22}
R_2(\mathbf{t},\mathbf{P})=t_{42}\log_{2}\left(1+\frac{P_{42}h_2}{N_0}\right).
\end{equation}
\section{Common Throughput Maximization}
In this paper, we focus on maximizing the minimum (max-min) throughput of the two users by jointly optimizing the time allocated to the HAP, WD$_1$ and WD$_2$ ($\mathbf{t}$), and power allocation $\mathbf{P}$, i.e.,
\begin{equation}
\label{1}
   \begin{aligned}
    (\rm{P1}): & \max_{\mathbf{t},\mathbf{P}} & &  \min(R_1(\mathbf{t},\mathbf{P}),R_2(\mathbf{t},\mathbf{P}))\\
    &\text{s. t.}    & &(\ref{t}) , (\ref{P3}),\rm{and}  (\ref{con}),\\
    & & & t_1,t_2,t_3,t_{41},t_{42}\geq 0,\\
    & & & P_2,P_3,P_{41},P_{42}\geq 0.
   \end{aligned}
\end{equation}
Noticed that if we set $t_2=0,t_{41}=0$ and $P_2=0,P_{41}=0$. Then (P1) reduces to the special case of  WPCN without cooperation, i.e., the near user WD$_2$ does not help the far user WD$_1$ with relaying its information to the HAP.

 (P1) is non-convex in the above form due to the multiplicative terms in (\ref{con}). To transform (P1) into a convex problem, we introduce auxiliary  variables $\tau_{41}=t_{41}P_{41}$ and $\tau_{42}=t_{42}P_{42}$. With $P_3$ in (\ref{P3}), $R_1^{(2)}(\mathbf{t},\mathbf{P}), R_1^{(3)}(\mathbf{t},\mathbf{P}), R_1^{(4)}(\mathbf{t},\mathbf{P})$ in (\ref{Ry12})-(\ref{Ry14}) can be re-expressed as function of $\mathbf{t}$, and $R_2(\mathbf{t},\mathbf{P})$ in (\ref{22}) can be re-expressed as function of $\mathbf{t}$ and $\boldsymbol{\tau}=[\tau_{41},\tau_{42}]$, i.e.,
\begin{equation}
\label{Ry122}
R_1^{(2)}(\mathbf{t})=t_3\log_{2}\left(1+{\rho_1^{(2)}}{\frac{t_1}{t_3}}\right),
R_1^{(3)}(\mathbf{t})=t_3\log_{2}\left(1+{\rho_1^{(3)}}{\frac{t_1}{t_3}}\right),
\end{equation}
\begin{equation}
\label{Ry144}
R_1^{(4)}(\mathbf{t},\boldsymbol{\tau})=t_{41}\log_{2}\left(1+{\rho_2}{\frac{\tau_{41}}{t_{41}}}\right),
R_2(\mathbf{t},\boldsymbol{\tau})=t_{42}\log_{2}\left(1+{\rho_2}{\frac{\tau_{42}}{t_{42}}}\right),
\end{equation}
where $\rho_1^{(2)}=h_1h_{12}{\frac{\eta P_1}{N_0}}$, $\rho_1^{(3)}={h_1^2{\frac{\eta P_1}{N_0}}}$, $\rho_2=\frac{h_2}{N_0}$ are constant parameters.

Accordingly, by introducing another auxiliary variable $\bar{R}$, (P$_1$) can be equivalently transformed into the following epigraph form:
\begin{equation}
\label{2}
 \begin{aligned}
    (\rm{P2}): & \max_{\overline{R},\mathbf{t},\boldsymbol{\tau}} & &  \overline{R} \\
    &\text{s. t.}    & & t_0+t_1+t_2+t_3+t_{41}+t_{42}\leq 1, \\
    & & & \tau_{41}+\tau_{42}\leq E_2^{(1)}+E_2^{(2)},\\
    & & & \overline{R}\leq R_1^{(1)}(\mathbf{t})+R_1^{(2)}(\mathbf{t}),\\
    & & & \overline{R}\leq R_1^{(3)}(\mathbf{t})+R_1^{(4)}(\mathbf{t},\boldsymbol{\tau}),\\
    & & & \overline{R}\leq R_2(\mathbf{t},\boldsymbol{\tau}).
 \end{aligned}
\end{equation}
Notice that $R_1^{(2)}(\mathbf{t}),R_1^{(3)}(\mathbf{t}),R_1^{(4)}(\mathbf{t},\boldsymbol{\tau})$ and $R_2(\mathbf{t},\boldsymbol{\tau})$ are all concave functions (see the proof in \cite{2014:Ju2}), therefore (P2) is a convex optimization problem, which can be easily solved by off-the-shelf convex optimization algorithms, e.g., interior point method. Then, after obtaining the optimal $\boldsymbol{\tau}^*$ and $\mathbf{t}^*$ in (P2), the optimal $\mathbf{P}^*$ in (P1) can be easily retrieved by setting $P_{41}^* = \tau_{41}^*/t_{41}^*$ and  $P_{42}^* = \tau_{42}^*/t_{42}^*$.

\section{Simulation Results}

In this section, we use simulations to evaluate the performance of the proposed cooperation method. In all simulations, we use the parameters of Powercast TX91501-1W transmitter as the energy transmitter at the HAP and those of P2110 Power harvester as the energy receiver at each WD with $\eta=0.6$ energy harvesting efficiency. Without loss of generality, it is assumed that the noise power is set $N_0=10^{-10}$ W for all receivers, the introduced additional noise power for ID circuit in Fig.~\ref{Fig.3} is $N_s=10^{-10}$ W. The channel gain $h_i=G_A(\frac{3\times10^8}{4\pi d_i f_c})^\lambda$, where $d$ denotes the distance separation between two devices, e.g., HAP-to-WD distance or the distance between the two WDs. $G_A=2$ denotes the antenna power gain, $\lambda=2$ denotes the path-loss factor, $\beta$ =0.8 denotes the power splitting factor, $\mu=0.8$ is set as a fixed backscatter reflection coefficient and $f_c=915$ MHz denotes the carrier frequency.

Fig.~\ref{Fig.4} compares the achievable max-min throughput of different schemes when the inter-user channel $h_{12}$ varies. In this case, the HAP and the two users are assumed to lie on a straight line in which the near user WD$_2$ is in the middle with $d_{12}=d_1-d_2$. Here, we fix $d_2=3$ meters and vary $d_1$ from 6 to 10 meters. Besides, we consider two different AB communication rates $R_b=5$ kbps, 50 kbps. Evidently, the throughput performance decreases with $d_1$ for all the methods due to the worse inter-user channel $h_{12}$. Besides, both user cooperation methods, either with or without AB communication, outperforms the independent transmission scheme. For the two cooperation methods, when $R_b = 50$ kbps, the proposed AB-assisted cooperation outperforms the one without AB communication when $d_1>6.8$, but produces worse performance otherwise. Similar result is also observed when $R_b =5$ kbps, where the proposed method has better performance when the inter-user channel is relatively weak. This is because when the far user WD$_1$ moves more away from the HAP, it suffers from more severe attenuation in both energy harvesting and information transmission to WD$_2$. Therefore, the optimal solution allocates more time to both WET and information exchange from WD$_1$ to WD$_2$ if AB communication is not used. The application of AB communication can effectively reduce the energy and time consumed on information exchange, thus can improve the overall throughput performance.
\begin{figure}[t]
\vspace{-0.4cm}
\setlength{\abovecaptionskip}{-0.1cm}
\setlength{\belowcaptionskip}{-1cm}
  \centering
   \begin{center}
      \includegraphics[width=0.7\textwidth]{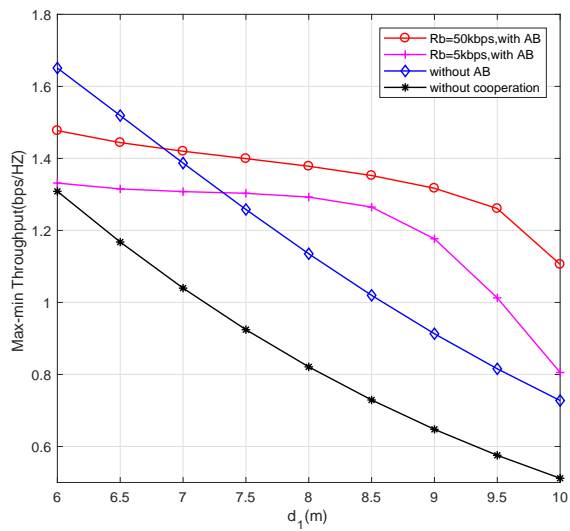}
   \end{center}
  \caption{The impact of inter-user channel ($h_{12}$) to the optimal throughput performance}
  \label{Fig.4}
\end{figure}
\begin{figure}
\vspace{-0.6cm}
\setlength{\abovecaptionskip}{-0.2cm}
\setlength{\belowcaptionskip}{-1cm}
  \centering
   \begin{center}
      \includegraphics[width=0.7\textwidth]{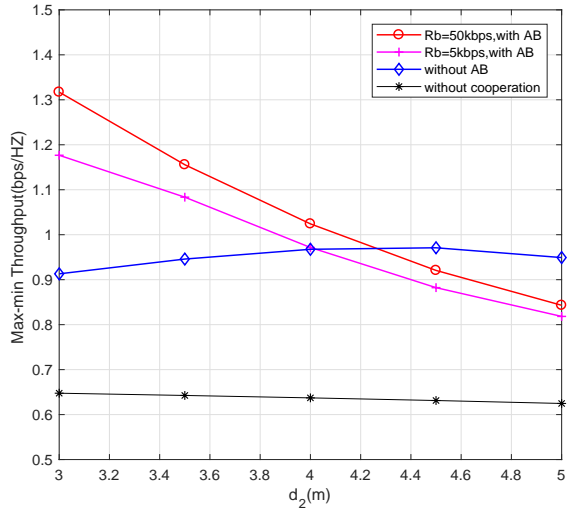}
   \end{center}
  \caption{The impact of relaying channel ($h_2$) to the optimal throughput performance }
  \label{Fig.5}
\end{figure}

Fig.~\ref{Fig.5} shows the impact of the HAP-to-WD$_2$ (relaying) channel to the optimal throughput performance. Here, We set $d_1=9$ meters, and vary $d_2$ from 3 to 5 meters. Noticed that the performance of non-cooperation scheme hardly changes as $d_2$ increases, this is because its throughput is mainly constrained by the weak channel between the far user WD$_1$ to HAP. It is observed that the proposed user cooperation has better performance than the one without AB communication when the helping relay is close to the HAP ($d_2$ is small). Again, this is because when $d_2$ is small, the separation between the two WDs is large thus the inter-user channel is weak. Therefore, WD$_1$ needs to consume significant amount of energy if transmitting actively to the helping WD. The use of AB-assisted cooperation can effectively reduce the energy consumptions and thus improve the throughput performance. The simulation results in Fig.~\ref{Fig.4} and Fig.~\ref{Fig.5} demonstrate the advantage of applying AB communication to improve the throughput performance of user cooperation in WPCN under various practical setups, especially when the inter-user channel is relatively weak.

\section{Conclusions}
In this paper, we proposed a novel user cooperation method using AB communication in a two-user WPCN. In particular, we studied the maximum common throughput optimization problem of the proposed model, and proposed efficient method to obtain the optimal solution. By comparing with representative benchmark methods, we showed that the proposed AB-assisted cooperation can effectively improve the throughput fairness performance in WPCN.

\end{document}